\documentclass[journal]{IEEEtran}
\usepackage{algorithm, algpseudocode, amsthm, amsfonts, datetime, mathtools, graphicx, amsmath, amssymb, bm, amsbsy, breqn, accents, float, listings, courier, lscape, multirow, multicol, longtable, dsfont, bbold, cite, soul, color, tikz, scalerel,empheq, paralist, cleveref, enumitem, subcaption,bm}

\newcommand{\R}{\mathbb{R}}

\newcommand{\I}{\mathbb{I}}

\newcommand{\indicator}{\mathbb{1}}
\newcommand{\expect}{\mathbb{E}}
\newcommand{\Expect}[1]{\expect\bracketSquare{#1}}
\newcommand{\Var}[1]{\Sigma_{#1}}

\newcommand{\scrN}{\mathcal{N}}

\newcommand{\bracketRound}[1]{\left(#1\right)}
\newcommand{\bracketSquare}[1]{\left[#1\right]}

\DeclareMathOperator*{\argmin}{arg\,min}
\DeclareMathOperator*{\argmax}{arg\,max}

\newtheorem{theorem}{Theorem}

\newtheorem{lemma}{Lemma}
\newtheorem{definition}{Definition}

\DeclareMathOperator{\prob}{Pr}
\newcommand{\Prob}[1]{\prob{\bracketRound{#1}}}

\newcommand{\trans}{^\intercal} 

\newcommand{\nUsers}{p}
\newcommand{\nNbr}{q}
\newcommand{\nMal}{r}

\newcommand{\numData}{N}
\newcommand{\dataRV}{G}

\newcommand{\linearTransformation}{Q}
\newcommand{\metric}{d}
\newcommand{\dimFeature}{m}
\newcommand{\dimReducedFeature}{k}
\newcommand{\changePointExact}{\tau}

\newcommand{\graph}{G}
\newcommand{\graphSpace}{\mathcal{G}}
\newcommand{\vertices}{V}
\newcommand{\edges}{E}
\newcommand{\adjacencyMatrix}{A}
\newcommand{\maliciousUsers}{v_{\nUsers-\nMal-1},v_{\nUsers-\nMal},\ldots,v_{\nUsers}\in\vertices}
\newcommand{\distributionEclipseAttackAbsent}{P_1}
\newcommand{\distributionEclipseAttackPresent}{P_2}
\newcommand{\brownianBridge}{\mathcal{B}}
\newcommand{\dataset}{\dataRV_i\in\graphSpace,\;i=1,2,\ldots,\numData}
\newcommand{\adjacencyMatrixDataset}{\adjacencyMatrix_{\dataRV_i},\;i=1,2,\ldots,\numData}
\newcommand{\processedAdjacencyMatrix}{\tilde{\adjacencyMatrix}}
\newcommand{\processedAdjacencyMatrixDataset}{\processedAdjacencyMatrix_{\dataRV_i},\;i=1,2,\ldots,\numData}
\newcommand{\testStat}{S}
\newcommand{\scaledTestStat}{T}
\newcommand{\changePointGuess}{n}
\newcommand{\boundary}{\delta}

\date{}
\begin{document}
\title{Eclipse Attack Detection on a Blockchain Network as a Non-Parametric Change Detection Problem}
\author{Anurag~Gupta\thanks{Anurag Gupta is with the School of Electrical \& Computer Engineering, Cornell University, Ithaca NY, 14853, USA.  (e-mail: anuragg.in@gmail.com).}  and
  Vikram Krishnamurthy\thanks{Vikram Krishnamurthy  is with the School of Electrical \& Computer Engineering, Cornell University, Ithaca NY, 14853, USA.  (e-mail: vikramk@ece.cornell.edu).}
  and Brian~Sadler\thanks{Brian Sadler is with DEVCOM Army Research Laboratory, Adelphi, Maryland, U.S. (e-mail: Brian.sadler@ieee.org)}}
\maketitle
\begin{abstract}
    This paper introduces a novel non-parametric change detection algorithm to identify eclipse attacks on a blockchain network; the non-parametric algorithm relies only on the empirical mean and variance of the dataset, making it highly adaptable. An eclipse attack occurs when malicious actors isolate blockchain users, disrupting their ability to reach consensus with the broader network, thereby distorting their local copy of the ledger. To detect an eclipse attack, we monitor changes in the Fréchet mean and variance of the evolving blockchain communication network connecting blockchain users. First, we leverage the Johnson-Lindenstrauss lemma to project large-dimensional networks into a lower-dimensional space, preserving essential statistical properties. Subsequently, we employ a non-parametric change detection procedure, leading to a test statistic that converges weakly to a Brownian bridge process in the absence of an eclipse attack. This enables us to quantify the false alarm rate of the detector. Our detector can be implemented as a smart contract on the blockchain, offering a tamper-proof and reliable solution. Finally, we use numerical examples to compare the proposed eclipse attack detector with a detector based on the random forest model. 
\end{abstract}


\IEEEpeerreviewmaketitle
\section{Introduction}
Blockchain, an immutable ledger distributed across multiple users~\cite{2017:MN-et-al}, relies on consensus among its users to share data. This paper  studies adversarial attacks on blockchains, with a specific focus on eclipse attacks~\cite{2015:EH-et-al}.  In an eclipse attack, malicious users isolate a victim user, disrupting their ability to reach a consensus with the rest of the network. 
 For example, if a user has eight incoming connections from other users, and an attacker controls all eight of those nodes, the attacker can refuse to relay any new blocks that rest of the network produce.
 Hence, detecting eclipse attacks are crucial for safeguarding blockchain networks. 
\subsubsection*{Main Results and Organization} 
To detect an eclipse attack, we propose a non-parametric change detection algorithm  that identifies  changes in the Fréchet mean and variance~\cite{2020:PD-HM} (these are  topological generalizations of mathematical expectation and variance) within a sequence of randomly evolving blockchain communication networks (BCNs). 
We exploit  the Johnson-Lindenstrauss (JL) lemma~\cite{2008:JM} to extract essential features from the large-dimensional BCN, ensuring that the test statistic is approximately preserved. In blockchain, a smart contract is a computer program that automatically executes a task 
based on a pre-specified conditions. Our proposed detector can be implemented as a smart contract on blockchain to detect an eclipse attack using a network monitor; this information can then be relayed to the blockchain users.

 Sec.\ref{sec:model}  formulates eclipse attack detection as a change detection problem on a space of directed graph and 
 describes  our proposed detector.
In Sec.\ref{sec:analysis-algorithm},  we analyze the performance of the detector using weak convergence methods. Specifically, Theorem~\ref{thm:null-hypothesis-results} shows that the scaled detector statistic converges weakly to a Brownian bridge process. As a result we can explicitly determine the false alarm by calculating the quantiles of a Brownian bridge. In the presence of an eclipse attack, Theorem~\ref{thm:alternate-hypothesis} estimates the onset of the eclipse attack. Finally, Theorem~\ref{thm:jl-lemma-test-statistic} shows the effect of 
 the JL lemma on the false positive alarm rate of the detector.

 Sec.\ref{sec:simulations}  assesses the performance of our eclipse attack detector using numerical examples on simulated datasets. 
We also provide numerical examples comparing the proposed eclipse attack detector and a detector based on the random forest model (RFM).


	\subsection*{Related Works}
	In the literature, several detectors have been proposed for detecting eclipse attacks. \cite{2023:DB-RB} and \cite{2020:GX-et-al} utilize random forest classification to analyze communication traffic and train their models on eclipse attack datasets. \cite{2022:QD-et-al} employs deep learning technique for detecting eclipse attacks. \cite{2020:BA-et-al} uses the blockchain's block creation rate as a detection metric. \cite{2021:HZ-et-al} monitors change in the proof-of-work difficulty levels to identify eclipse attacks. 
    
    Related to attack mitigation, a peer selection strategy introduced by \cite{2022:AY-et-al} offers a way to reduce the likelihood of eclipse attacks. Eclipse attacks share similarities with Sybil attacks~\cite{2021:MI-RM} and routing attacks~\cite{2022:RC-et-al}, both of which can impact the integrity of the blockchain consensus protocol.
    
    Our eclipse attack detection approach distinguishes itself by not requiring training data. Instead, we employ statistical tools from \cite{2020:PD-HM}, which offer a generalized solution for change detection in arbitrary object spaces. We identify eclipse attacks by tracking changes in the Fréchet mean and variance within the sequence of randomly evolving BCN. 
	
	\section{Detecting Eclipse Attack on a Blockchain Network}
 \label{sec:model}
	In this section, we formulate  detecting eclipse attacks on a blockchain network as a change detection problem and present our detection algorithm. 
    \subsection{Model for Eclipse Attack}
    
        
        
    We begin by modeling the BCN as a directed graph and defining its adjacency matrix.
	\begin{definition}
		\label{def:communication-network}
		A BCN is represented as a directed graph $\graph=(\vertices,\edges)\in\graphSpace$, where $\graphSpace$ is the graph space comprising $\nUsers$ vertices. Each vertex has $\nNbr$ outgoing edges.
	\end{definition}
	The adjacency matrix $\adjacencyMatrix_{\graph}\in\R^{|\vertices|\times|\vertices|}$ of the BCN $\graph$ is defined as follows:
	\begin{align}
	\label{eq:adjacency-matrix}
	\adjacencyMatrix_{\graph}(i,j)=\begin{cases}
	1, & \text{$\exists$ an edge from the vertex $j$}\\
	&\text{to the vertex $i$} \\
	0, & \text{otherwise}
	\end{cases}
	\end{align}
	
        Consensus in blockchain relies on peer-to-peer (P2P) communication. The BCN serves to illustrate the flow of information among blockchain users.
 In the absence of an eclipse attack, the BCN at each time $t$ resembles a random graph with uniform distribution. Here, each user simply selects $\nNbr$ neighbors in a random and uniform manner to share information. However, during an eclipse attack, malicious users target victim users with substantial computational power. These malicious actors choose their neighbors in a non-uniform manner to disrupt the consensus of the victim users. 
    
    In this work, we assume: (1) When all users select their neighbour honestly using the blockchain communication protocol, the random BCN $\graph$ follows an unknown but fixed distribution   $\distributionEclipseAttackAbsent$. (2) The eclipse attack strategy is time-invariant\footnote{For an eclipse attack, multiple malicious users must communicate continuously with the victim user(s). A complex eclipse attack strategy can slow the communication rate and make the attack ineffective. Therefore, the assumption of a time-invariant eclipse attack strategy is justified.}. Consequently, in the presence of an eclipse attack, the BCN $\graph$ follows an unknown but fixed distribution $\distributionEclipseAttackPresent$. For instance, a common eclipse attack strategy employed by malicious users is to choose the victim users as their neighbors with a significantly higher probability compared to other users. Now, let's provide a formal definition of our model for the eclipse attack.
	
	\begin{definition}[Eclipse attack]
		\label{def:eclipse-attack}
		A blockchain is free from an eclipse attack if the random BCN, as represented by the graph $\graph$ (Definition~\ref{def:communication-network}), is sampled from $\graphSpace$ following the distribution $\distributionEclipseAttackAbsent$. Conversely, a blockchain is under an eclipse attack if the random graph $\graph$ is sampled from $\graphSpace$ following the distribution $\distributionEclipseAttackPresent$.
	\end{definition}
	\textit{Example.} Consider a blockchain network with $\nUsers$ blockchain users, each of whom selects $\nNbr$ neighbors for consensus. An example of the distribution $\distributionEclipseAttackAbsent$ is when each of the neighbors in the BCN is selected uniformly at random, i.e.,
	\begin{align*}
	\Prob{\adjacencyMatrix_{\graph}(i,j)=1}=\frac{\nNbr}{\nUsers}\quad\text{s.t. }\sum_{i}\adjacencyMatrix_{\graph}(i,j)=\nNbr, \forall j
	\end{align*}
	where, $\nNbr$ is defined in Definition~\ref{def:communication-network}.
	
	Now, let's consider an example of the distribution $\distributionEclipseAttackPresent$. Here, $\nMal$ malicious blockchain users, denoted as $\maliciousUsers$, choose $v_1\in\vertices$ as their victim. 
    
    For $j=\maliciousUsers$
	\begin{align*}
	\Prob{\adjacencyMatrix_{\graph}(i,j)=1}&>\frac{\nNbr}{\nUsers},\quad i=v_1\\
	\Prob{\adjacencyMatrix_{\graph}(i,j)=1}&<\frac{\nNbr}{\nUsers},\quad \text{otherwise}
    \end{align*}
    For $j\neq \maliciousUsers$
    \begin{align*} \Prob{\adjacencyMatrix_{\graph}(i,j)=1}&=\frac{\nNbr}{\nUsers}\quad\text{s.t. }\sum_{i}\adjacencyMatrix_{\graph}(i,j)=\nNbr, \forall j
	\end{align*}
    
    In this example, the victim user heavily depends on information provided by attackers to keep up with the current state of the blockchain. As a result, the victim's local copy of the distributed ledger no longer aligns with the majority consensus of the blockchain network.

    \subsection{Eclipse Attack Detection Problem}
	We now formulate the eclipse attack detection problem as a change detection problem. The proposed  detector operates on an offline dataset of BCNs\footnote{Note that changes in the BCN occurs at a faster rate than the addition of new blocks to the blockchain. This allows us to observe a substantial sample of BCNs before a double spend attack resulting from an eclipse attack is achieved. Therefore, using offline datasets for eclipse attack detection is practical.}; the BCN can be monitored using a network monitor. We formulate the eclipse attack detection problem as a hypothesis testing problem.
	\begin{definition}[Eclipse attack detection problem]
		\label{def:hypothesis-testing}
		Let the sequence of random graphs $\{\dataset\}$ denote the sequence of BCNs observed. The eclipse attack detection problem on a blockchain network is the following hypothesis testing problem
		\begin{align}
		\label{eq:hypothesis}
		\begin{aligned}
		H_0&: \dataRV_1, \dataRV_2, \ldots, \dataRV_\numData \sim \distributionEclipseAttackAbsent\\
		H_1&: \exists\; \changePointExact\in\{1,\ldots,\numData\} \\&s.t.\: \left\{\begin{array}{l} \dataRV_1, \dataRV_2, \ldots, \dataRV_{\changePointExact-1} \sim \distributionEclipseAttackAbsent \\ \dataRV_{\changePointExact}, \dataRV_{\changePointExact+1}, \ldots, \dataRV_\numData \sim \distributionEclipseAttackPresent\end{array}\right.
		\end{aligned}
		\end{align}
	\end{definition}
	Here, $\changePointExact$ denotes the the onset of the eclipse attack on the blockchain network. The eclipse attack detection problem~\eqref{eq:hypothesis} is a change detection problem on a space of directed graphs.

    \subsection{Test Statistic for Detecting Eclipse Attack}
    \label{sec:test-statistic}
	In this section, we present a test statistic to solve the eclipse attack detection problem~\eqref{eq:hypothesis}. The proposed test statistic estimates changes in the mean and variance of the sequence of BCNs. However, the communication network do not lie in the Euclidean space. So, we use the concept of Fréchet mean and variance~\cite{2020:PD-HM}, a topological generalization of mean and variance\footnote{Fréchet mean $\mu$ and Fréchet variance $V$ of a probability measure $P$ is defined as follow:
		\begin{align*}		
		\mu=\arg\min_{\omega \in \graphSpace} \Expect{\metric^2(\dataRV, \omega)},\quad
		V=\min _{\omega \in \graphSpace} \Expect{\metric^2(\dataRV, \omega)}
		\end{align*}
		Here, $\dataRV\sim P$ denotes a random object with probability measure $P$; $\graphSpace$ denotes the sample space of the random object $\graph$; and $\metric$ denotes a suitable choice of distance metric on the space $\graphSpace$.}.
    To calculate the Fréchet mean and variance, we define a distance metric $\metric$ on the space of the BCN $\graphSpace$ (Definition \ref{def:communication-network}). This metric measures the dissimilarity between two BCNs $\graph_1$ and $\graph_2$ using the Frobenius norm.
	\begin{definition}
		\label{def:distance-communication-network}
		The distance $\metric$ between two BCNs $\graph_1,\graph_2\in \graphSpace$, (Definition~\ref{def:communication-network}), is defined as the Frobenius norm\footnote{Our model assumes that the number of users in the blockchain is fixed. Hence, distance between two BCNs is well-defined.} of the difference between their adjacency matrices~\eqref{eq:adjacency-matrix}.
		\begin{align}
		\label{eq:distance-metric}
		\metric(\graph_1,\graph_2)&=\left(\sum_{i,j}|\adjacencyMatrix_{\graph_1}(i,j)-\adjacencyMatrix_{\graph_2}(i,j)|^2\right)^{\frac{1}{2}}
		\end{align}
	\end{definition}
	
 The test statistic partitions the sequence of BCNs into two parts. The goal is to determine if the BCNs in these two components are sampled from the same or distinct distributions. To achieve this, the test statistic examines the Fréchet mean and variance of the BCNs in each component.

Under the null hypothesis $H_0$~\eqref{eq:hypothesis}, i.e., absence of an eclipse attack on the blockchain network, the Fréchet mean and variance of the BCNs in both parts are the same. Conversely, under the alternate hypothesis $H_1$~\eqref{eq:hypothesis}, the Fréchet mean and variance of the BCNs in these parts differ, signaling the presence of an eclipse attack.

Before introducing our test statistic for eclipse attack detection, we define several mathematical quantities that rely on the adjacency matrices $\adjacencyMatrixDataset$ of the sequence of BCNs. We also introduce the term $\changePointGuess$, which represents an estimate for the change point $\changePointExact$ in the eclipse attack detection problem~\eqref{eq:hypothesis}. For each $\changePointGuess\in\{1,\ldots, N-1\}$, we proceed to define these quantities and present our test statistic.
    \begin{align}
    \label{eq:test-statistic-predefinition}
	\nonumber\hat{\mu}_{\changePointGuess}&:=\arg\min_{\omega \in \graphSpace} \frac{1}{\changePointGuess} \sum_{i=1}^{\changePointGuess} \metric^2\left(\graph_{i}, \graph_{\omega}\right) \\
    \nonumber\hat{V}_{\changePointGuess}&:=\frac{1}{\changePointGuess} \sum_{i=1}^{\changePointGuess} \metric^2\left(\graph_{i}, \hat{\mu}_{\changePointGuess}\right)\\
    \nonumber\hat{\mu}_{\numData-\changePointGuess}&:=\arg\min_{\omega \in \graphSpace} \frac{1}{(\numData-\changePointGuess)} \sum_{i=\changePointGuess+1}^\numData \metric^2\left(\graph_{i}, \graph_{\omega}\right)\\
    \nonumber\hat{V}_{\numData-\changePointGuess}&:=\frac{1}{(\numData-\changePointGuess)} \sum_{i=\changePointGuess+1}^\numData \metric^2\left(\graph_{i}, \hat{\mu}_{\numData-\changePointGuess}\right)\\
    \nonumber\hat{V}_{\changePointGuess}^C&:=\frac{1}{\changePointGuess} \sum_{i=1}^{\changePointGuess} \metric^2\left(\graph_{i}, \hat{\mu}_{\numData-\changePointGuess}\right)\\
    \nonumber\hat{V}_{\numData-\changePointGuess}^C&:=\frac{1}{\numData-\changePointGuess} \sum_{i=\changePointGuess+1}^\numData \metric^2\left(\graph_{i}, \hat{\mu}_{\changePointGuess}\right)\\
    \nonumber\hat{\mu}&:=\arg\min_{\omega \in \graphSpace} \frac{1}{\numData} \sum_{i=1}^\numData \metric^2\left(\graph_{i}, \graph_{\omega}\right)\\
    \nonumber\hat{V}&:=\frac{1}{\numData} \sum_{i=1}^\numData \metric^2\left(\graph_{i}, \hat{\mu}\right)\\
    \hat{\sigma}^2&:=\frac{1}{\numData}\left[ \sum_{i=1}^\numData \metric^4\left(\graph_{i}, \hat{\mu}\right)-\hat{V}^2\right]
    \end{align}
    The test statistic compares the Fréchet mean and variance of the BCNs $\graph_1,\graph_2,\ldots,\graph_\changePointGuess$ and $\graph_{\changePointGuess+1},\graph_{\changePointGuess+2},\ldots\graph_{\numData}$. 
	\begin{definition}[Test statistic for detecting an eclipse attack]
		\label{def:test-statistic}
		Let $\changePointGuess$ denote the estimate for the change point $\tau$ in the eclipse attack detection problem~\eqref{eq:hypothesis}. The test statistic $\testStat_{\changePointGuess,\numData}$ is defined as follows:
		\begin{align}
		\label{eq:statistic-definition}
		\begin{aligned}
		\testStat_{\changePointGuess,\numData}&=\frac{\changePointGuess(\numData-\changePointGuess)}{\numData^2\hat{\sigma}^2}\left\{\left(\hat{V}_{\changePointGuess}-\hat{V}_{\numData-\changePointGuess}\right)^2+\right.\\
		&\left.\left(\hat{V}_{\changePointGuess}^C-\hat{V}_{\changePointGuess}+\hat{V}_{\numData-\changePointGuess}^C-\hat{V}_{\numData-\changePointGuess}\right)^2\right\}
		\end{aligned}
		\end{align}		
		Here, 
  $\hat{V}_{\changePointGuess},\hat{V}_{\numData-\changePointGuess}^C\hat{V}_{\changePointGuess}^C, \hat{V}_{\numData-\changePointGuess}^C,\hat{\sigma}^2$ are defined in~\eqref{eq:test-statistic-predefinition}.
	\end{definition}
	The test statistic $\testStat_{\changePointGuess,\numData}$ in~\eqref{eq:statistic-definition} comprises two terms: the first term estimates the change in Fréchet variance, while the second term estimates the change in the Fréchet mean of the BCNs in the two components.

    \subsection{Dimensionality Reduction of the Adjacency Matrices}
	In this section, we use the JL lemma to to reduce the dimension of the adjacency matrix. Remember that the proposed test statistic~\eqref{def:test-statistic} is computed using the sequence of adjacency matrices for the BCNs. The number of elements in the adjacency matrix grows as the square of the number of blockchain users. Hence, it is necessary to reduce the dimension of the adjacency matrices $\adjacencyMatrix_\graph$ to decrease the computational cost of the test statistic $\testStat_{\changePointGuess,\numData}$~\eqref{eq:statistic-definition}. 
 In this work, we leverage the JL lemma to project the adjacency matrices of BCNs into a lower-dimensional subspace while approximately preserving the test statistic.
	\begin{lemma}[Johnson-Lindenstrauss (JL) lemma]
		\label{lemma:jl}
		Given any $\epsilon \in (0, 1)$ and an integer $\numData$, let $\dimReducedFeature$ be a positive integer satisfying
		$\dimReducedFeature \geq \frac{24}{3 \epsilon^2 - 2 \epsilon^3} \log \numData$.
		For any set $A$ containing $\numData$ points in $\R^\dimFeature$, there exists a mapping $f: \R^\dimFeature \rightarrow \R^\dimReducedFeature$ such that for all $x, y \in A$, the following inequality holds:
		$(1-\epsilon)\|x-y\|^2 \leq \|f(x)-f(y)\|^2 \leq (1+\epsilon)\|x-y\|^2$
	\end{lemma}
	The linear map $f$ in Lemma~\ref{lemma:jl} can be found using random projections in randomized polynomial time~\cite{1974-JG}.
	Now, we apply the JL lemma on the adjacency matrices  to obtain the projected adjacency matrices.
	
	\begin{definition}[Projected adjacency matrices]
		\label{def:processed-adjacency-matrix}
		The projected adjacency matrices, denoted as $\processedAdjacencyMatrixDataset$ for the BCNs $\dataset$, are obtained by applying the JL lemma (Lemma~\ref{lemma:jl}) to the adjacency matrices $\adjacencyMatrixDataset$~\eqref{eq:adjacency-matrix}, with an appropriately chosen value of $\epsilon$. Equivalently, \begin{align}
		\label{eq:processed-adjacency-matrix}
		\processedAdjacencyMatrix_{\graph_i}=f(\adjacencyMatrix_{\graph_i}),i=1,2,\ldots,\numData\end{align} 
		where the linear map $f$ satisfies Lemma~\ref{lemma:jl}.       
	\end{definition}
 
    \subsubsection*{Comparison between the Adjacency Matrix $\adjacencyMatrix_\graph$ of the BCN And The Projected Adjacency Matrix $\processedAdjacencyMatrix_\graph$}
		To apply the JL lemma (Lemma~\ref{lemma:jl}), we first vectorize the adjacency matrix $\adjacencyMatrix_\graph$. Denote the vectorized $\adjacencyMatrix_\graph$ as $X$. Then, we compute a suitable linear transformation $\linearTransformation$ that satisfies the JL lemma for the chosen value of $\epsilon$. Now,
		$Y=\linearTransformation X\Rightarrow \Expect{Y}=\linearTransformation\Expect{X}\Rightarrow \Var{Y}=\linearTransformation\Var{X}\linearTransformation\trans$.
		As the proposed non-parametric statistical detector detects a change in the mean and the variance of the adjacency matrices of the random BCNs, we can use the projected adjacency matrix $\processedAdjacencyMatrix_\graph$ to compute the test statistic. This is because the mean of the projected adjacency $\processedAdjacencyMatrix_\graph$ is a linear transformation of the mean of the adjacency matrix $\adjacencyMatrix_\graph$, and the variance of the projected adjacency matrix $\processedAdjacencyMatrix_\graph$ is similar to the variance of the adjacency matrix $\adjacencyMatrix_\graph$.
    
    \subsection{Algorithm for Detecting  Eclipse Attack}
 Having developed the necessary mathematical tools, we present the eclipse attack detection algorithm. Algorithm \ref{alg:change-detection} outlines the steps in this algorithm. Given the large dimension of the BCN, we initially employ the JL lemma to reduce dimensionality while approximately preserving the test statistic defined in~\eqref{eq:statistic-definition}. We also assume that the eclipse attack do not occur near the endpoints\footnote{We assume that the eclipse attack do not occur near the endpoints of the sequence of communication networks. To detect an eclipse attack near the endpoints, the detector can use an overlapping sequence of BCNs, ensuring that the attack takes place away from the endpoints for at least one batch. Alternatively, one can refine the test statistic to detect an eclipse attack near endpoints (as explored in~\cite{2020:LH-et-al}), a topic we plan to investigate in future research.
}, i.e., $\changePointExact\in\I^+,\frac{\changePointExact}{\numData}\in(\boundary,1-\boundary)$ for some $\boundary>0$.
	\begin{algorithm}[ht]
		\caption{Algorithm for detecting an eclipse attack on a blockchain network}
		\begin{algorithmic}[1]
			\Require Sequence of adjacency matrices $\adjacencyMatrixDataset$ of the random BCNs  $\dataset$ at time $t=1,2,\ldots,\numData$, respectively (Definition~\ref{def:hypothesis-testing}).
			\State \textbf{Dimensionality reduction:} Compute the projected adjacency matrices $\processedAdjacencyMatrixDataset$~\eqref{def:processed-adjacency-matrix} using the JL lemma.
			\State \textbf{Test statistic:} Compute the test statistic $\testStat_{\changePointGuess,\numData}$~\eqref{eq:statistic-definition} using the projected adjacency matrices $\processedAdjacencyMatrixDataset$ for $\changePointGuess=1,2,\ldots,\numData-1,\frac{\changePointGuess}{\numData}\in(\boundary,1-\boundary)$ for some $\boundary>0$. 
			\State \textbf{Asymptotic quantile:} Choose a level of significance $\alpha\in[0,1]$. Compute $q_{1-\alpha}=(1-\alpha)$ quantile of the distribution $\max_{\changePointGuess\in
\{1,2,\ldots,\numData-1\},\frac{\changePointGuess}{\numData}\in(\boundary,1-\boundary)}\brownianBridge^2\left(\frac{\changePointGuess}{\numData}\right)$. Here, $\brownianBridge(t)$
   is a Brownian bridge process on $[0,1]$ with the covariance function $C(t_1, t_2)=1$ for $0\leq t_1 \leq t_2\leq 1$.
			\If{$\max_{\changePointGuess\in\{1,2,\ldots\numData-1\},\frac{\changePointGuess}{\numData}\in(\boundary,1-\boundary)} \numData \testStat_{\changePointGuess,\numData}<q_{1-\alpha}$} 
			\Return No eclipse attack detected.
			\Else\quad\Return Eclipse attack detected  at time $$\changePointGuess^*:=\argmax_{\changePointGuess\in\{1,2,\ldots,\numData-1\},\frac{\changePointGuess}{\numData}\in(\boundary,1-\boundary)}\testStat_{\changePointGuess,\numData}$$
			\EndIf 
		\end{algorithmic}
		\label{alg:change-detection}
	\end{algorithm}	
	Let $k$ denote the dimension of the adjacency matrices. Then, the complexity of the Algorithm~\ref{alg:change-detection} is $O\left(\numData^2 k+ \numData|\graphSpace|\right)$, where $\graphSpace$ is defined in~\eqref{eq:adjacency-matrix}.
 
	To summarize, we designed an algorithm to detect an eclipse attack on a blockchain network. The proposed test statistic was based on Fréchet change detection~\cite{2020:PD-HM}. We also used the JL lemma to reduce the dimension of the BCNs.

 \section{Weak Convergence Analysis of Eclipse Attack Detector}
	\label{sec:analysis-algorithm}
	In this section, we analyze the test statistic for the proposed eclipse attack detector (Algorithm~\ref{alg:change-detection}). Under $H_0$ (absence of an eclipse attack), we prove that a scaled test statistic converges weakly to the square of a Brownian bridge process. Under $H_1$ (presence of an eclipse attack), we show that the peak of the test statistic estimates the onset of the eclipse attack.

    \subsection{Weak Convergence of Test Statistic}
	Our first result (Theorem~\ref{thm:null-hypothesis-results}) analyzes the asymptotics of the test statistic $\testStat_{\changePointGuess,\numData}$~\eqref{eq:statistic-definition} under the null hypothesis $H_0$~\eqref{eq:hypothesis}, i.e., absence of an eclipse attack on the blockchain network. Note that $\{\testStat_{\changePointGuess,\numData},\changePointGuess=1,2,\ldots,\numData-1\}$~\eqref{eq:statistic-definition}, represents a discrete-time stochastic process. 
 As is customary in weak convergence analysis~\cite{1984:HK,2009:SE-TK}, we first construct a continuous time stochastic process $\testStat_\numData(\numData t)$ by interpolating the discrete time test statistic process $\{\testStat_{\changePointGuess,\numData}\}$. 
  \begin{align}
        \label{eq:interpolation}\testStat_\numData(\numData t)&=\testStat_{\changePointGuess,\numData}\\
        \nonumber&\text{ for } \numData t\in [\changePointGuess,\changePointGuess+1),\;\changePointGuess=0,1,\ldots,\numData-1
    \end{align}
  The continuous time process $\testStat_\numData(\numData t)$ has sample paths in the function space $D[0,1]$, namely the space of functions that are continuous on the right with limit on the left (cadlag functions). 
   We define a scaled test statistic continuous time stochastic process $\scaledTestStat_\numData(t)$ as follows:
    \begin{align}
        \label{eq:scaled-test-stat}\scaledTestStat_\numData(t):=\numData\testStat_\numData(\numData t)
    \end{align}
 Theorem~\ref{thm:null-hypothesis-results} shows  that as $\numData\rightarrow \infty$, the scaled test statistic continuous time stochastic process
$\scaledTestStat_\numData(t) $ converges weakly (in Skorohod metric~\cite{2013:PB}) to the square of a Brownian bridge stochastic process.  Note that the weak convergence approach   deals  with the convergence of scaled sequences of the test statistic that are treated as stochastic processes rather than random variables. Thus, the weak convergence approach specifies convergence for  the entire trajectory of the test statistic of the detection algorithm.

	\begin{theorem}
		\label{thm:null-hypothesis-results}
		Assume that the eclipse attack do not occur near the endpoints, i.e., $\frac{\changePointExact}{\numData}\in[\boundary,(1-\boundary)]$ for some $\boundary>0$, where $\changePointExact$ is defined in~\eqref{eq:hypothesis}. Then, under $H_0$ (absence of an eclipse attack), the scaled test statistic~\eqref{eq:scaled-test-stat} process converges weakly:
    \begin{align*}
  \scaledTestStat_\numData(t)&\Rightarrow \brownianBridge^2(t)\\
  \intertext{Also, the continuous mapping theorem implies }\max_{t\in[\boundary,(1-\boundary)]}\scaledTestStat_\numData(t)&\Rightarrow \max_{t\in[\boundary,(1-\boundary)]}\brownianBridge^2(t)
		\end{align*}
		Here $\Rightarrow$ denotes weak convergence\footnote{
  Weak convergence in functional space is a generalization of the weak convergence in distribution for random variables.
  A sequence of probability measures $\mu_n$ converges weakly to the probability measure $\mu$ if, for all bounded and continuous test functionals $f$, the expected value of $f$ with respect to $\mu_n$ converges to the expected value of $f$ with respect to $\mu$, i.e., $\mathbb{E}_{\mu_n}[f] \rightarrow \mathbb{E}_{\mu}[f]$.
  }; 
  $\brownianBridge$ is a standardized Brownian bridge process\footnote{A standardized Brownian bridge on $[0,T]$ is a continuous-time stochastic process whose probability distribution is the conditional probability of the Weiner process $W(t)$ subject to the condition that $W(0)=W(T)=0$ with the covariance function $C(t_1, t_2)=1,\,0\leq t_1 \leq t_2\leq 1$.} with the covariance function $C(t_1, t_2)=1,\,0\leq t_1 \leq t_2\leq 1$.

	\end{theorem}
	\begin{proof}
	Appendix~\ref{sec:proof-null-hypothesis-results} of the supplementary material.
	\end{proof}
	
  Convergence to a Brownian bridge instead of Brownian process in Theorem~\ref{thm:null-hypothesis-results} is intuitive because $\scaledTestStat_\numData(t)\propto t(1-t)$. Theorem~\ref{thm:null-hypothesis-results} is used in steps 3-4 of Algorithm~\ref{alg:change-detection} to detect an eclipse attack. In practice, we declare the presence of an eclipse attack on a blockchain network if the maximum of the scaled test statistic exceeds the $0.95$ quantile, denoted as $q_{0.95}$, of the distribution $\max_{t \in [\boundary,1-\boundary]}\brownianBridge^2(t)$.
	
	Our second result 
 (Theorem~\ref{thm:alternate-hypothesis} below) 
 investigates the test statistic $\testStat_\numData(\numData t)$~\eqref{eq:statistic-definition} under $H_1$~\eqref{eq:hypothesis}, i.e., presence of an eclipse attack on the blockchain network. This result estimates the onset of the eclipse attack. Before presenting the theorem, we need to define the limiting test statistic $\testStat(\numData t)$:
\begin{align}
\label{eq:limit-test-statistic}
\testStat( t):=\lim_{\numData\rightarrow\infty}\testStat_\numData(\numData t)
\end{align}
Here, the test statistic $\testStat_\numData(\numData t)$ converges to $\testStat(t)$ in probability~\cite{2020:PD-HM}.
	\begin{theorem}
		\label{thm:alternate-hypothesis}
		Assume that the eclipse attack do not occur near the endpoints, i.e., $\frac{\changePointExact}{\numData}\in[\boundary,(1-\boundary)]$ for some $\boundary>0$, where $\changePointExact$ is defined in~\eqref{eq:hypothesis}. Then, under $H_1$ (presence of an eclipse attack), the maximum of the limiting test statistic $\testStat(\numData t)$ defined in~\eqref{eq:limit-test-statistic}  occurs at the onset, $\changePointExact$  of the eclipse attack:
        \begin{align*}
            \lim_{\numData\rightarrow\infty}\frac{\tau}{\numData}=\argmax_{t \in [\boundary,(1-\boundary)]}\testStat(t)
        \end{align*}
        Let $\changePointExact_\numData=\numData\arg\max_{t \in [\boundary,(1-\boundary)]}\testStat_\numData(\numData t)$ where $\testStat_\numData(\numData t)$ is defined in~\eqref{eq:interpolation}. Then for $\gamma>0$ the following holds
		\begin{align*}
		\begin{aligned}
		\Prob{\left|\frac{{\tau}_{\numData}-\tau}{\numData}\right|>\gamma}\rightarrow 0
	\end{aligned}
		\end{align*}        
	\end{theorem}
	\begin{proof}
        Appendix~\ref{sec:proof-alternative-hypothesis} of  supplementary material.
	\end{proof}
    The second statement of Theorem~\ref{thm:alternate-hypothesis} gives an error bound for estimating the onset of the eclipse attack using finite samples of BCNs. Theorem~\ref{thm:alternate-hypothesis} is used in step 5 of Algorithm~\ref{alg:change-detection} to estimate the onset of the eclipse attack
    using the discrete-time test statistic $\testStat_{\changePointGuess,\numData}$~\eqref{eq:interpolation} .
     \subsection{Effect of the Processed Adjacency Matrices on the Test Statistic}
	Our final result compares  the test statistic computed using the projected adjacency matrices and the original adjacency matrices of the BCN\footnote{
See Sec.\ref{sec:comparison-test-statistic} of the supplementary material for a numerical example illustrating Theorem~\ref{thm:jl-lemma-test-statistic}.}. It shows that the false positive alarm rate of the detector is higher when the test statistic is computed using the projected adjacency matrix $\processedAdjacencyMatrix_\graph$.
 \begin{theorem}
 	\label{thm:jl-lemma-test-statistic}
Let $\testStat_\numData(\numData t)$ defined in~\eqref{eq:interpolation} denote the test statistic computed using the original adjacency matrices~\eqref{eq:adjacency-matrix}. Let $\tilde{\testStat}_\numData(\numData t)$ denote the test statistic computed using the projected adjacency matrices~\eqref{eq:processed-adjacency-matrix}. Under $H_0$~\eqref{eq:hypothesis}, as $\numData\rightarrow\infty$, using projected adjacency matrices to compute the test statistic leads to a higher false positive alarm rate:
\begin{align*}
    \lim_{\numData\rightarrow\infty}\tilde{\testStat}_{\numData}(\numData t) \geq \lim_{\numData\rightarrow\infty}\testStat_{\numData}(\numData t)
\end{align*}
Here, the convergence to the limit is in probability. Furthermore,
\begin{align*}
    \lim_{\numData\rightarrow\infty} \tilde{\testStat}_\numData(\numData t) \geq \frac{5\epsilon\, t(1-t)V}{\hat{\sigma}^2}
\end{align*}
where $V=\lim_{\numData\rightarrow\infty}\hat{V}$ ($\hat{V}$ is defined in~\eqref{eq:test-statistic-predefinition}); $\epsilon\in(0,1)$; and $\hat{\sigma}^2$ is the empirical variance computed in~\eqref{eq:test-statistic-predefinition}.
\end{theorem}
	\begin{proof}
		Appendix~\ref{sec:proof-jl-lemma-test-statistic} of  supplementary material.
	\end{proof}
	
	In summary, we have presented three key results on the test statistic~\eqref{eq:statistic-definition} for detecting an eclipse attack: 1) Under the null hypothesis $H_0$~\eqref{eq:hypothesis}, the first result ensures weak convergence of the maximum of the scaled test statistic to the maximum of the square of the Brownian bridge process. 2) Under the alternate hypothesis $H_1$~\eqref{eq:hypothesis}, the second result estimates the onset of the eclipse attack on the blockchain network. 
 3) The third result investigates the impact on the false alarm rate of the detector when using the projected adjacency matrices~\eqref{eq:processed-adjacency-matrix} to compute the test statistic.
	\section{Numerical Examples}
	\label{sec:simulations}
	\begin{figure*}[t]
		\centering
		\begin{subfigure}[t]{.45\textwidth}
			\includegraphics[width=0.92\textwidth]{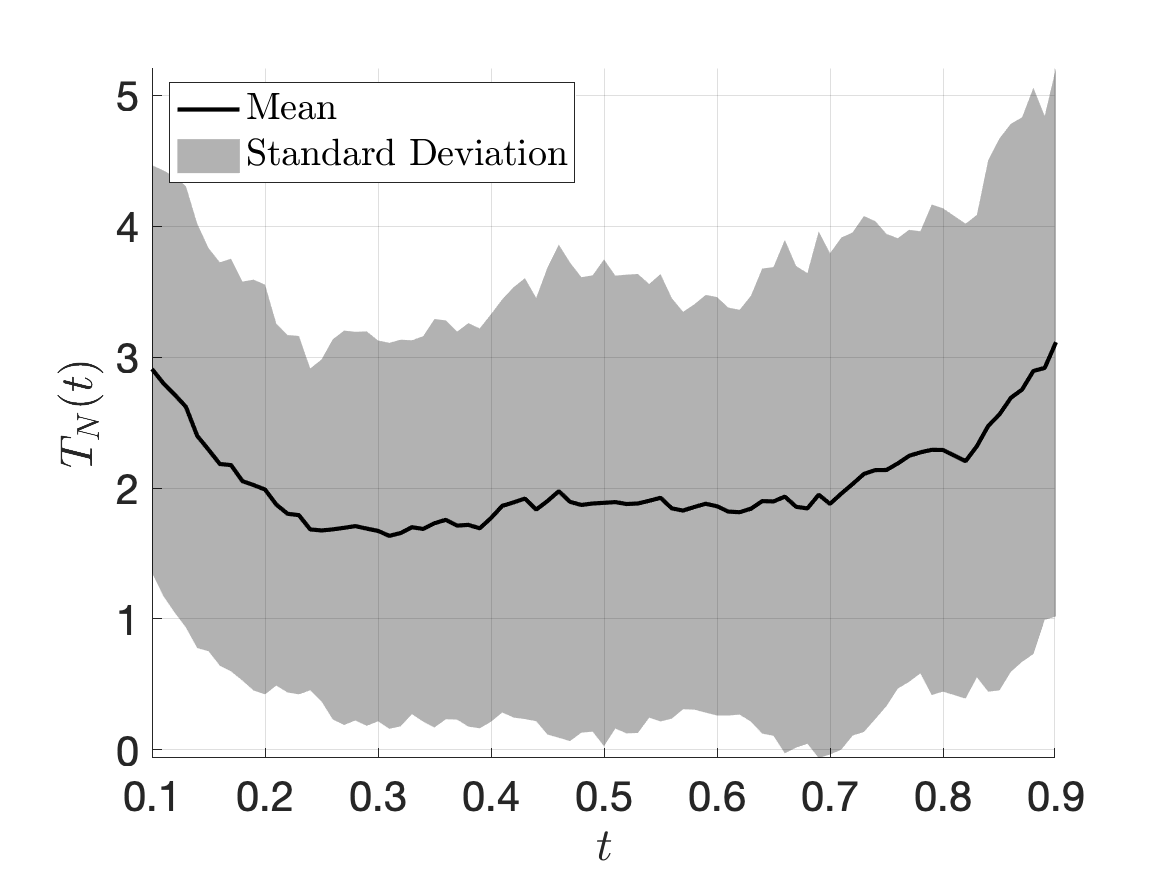}
			\caption{Absence of an eclipse attack.
   }    
		\end{subfigure}\hspace{0.2cm}
		\begin{subfigure}[t]{.45\textwidth}
			\includegraphics[width=0.92\textwidth]{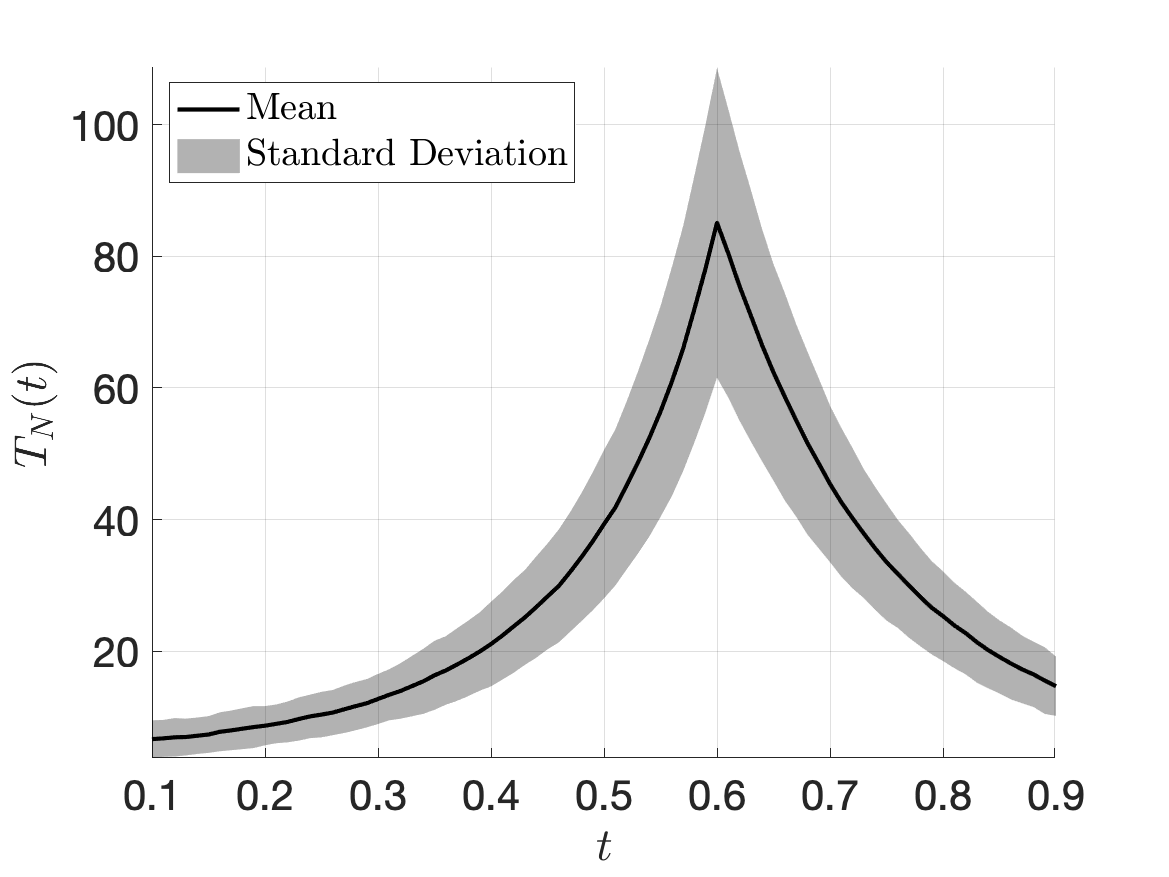}
			\caption{Presence of an eclipse attack at $\frac{\changePointExact}{\numData}=0.6$.
			}    
		\end{subfigure}
		\caption{Scaled test statistic $\scaledTestStat_\numData( t)$~\eqref{eq:scaled-test-stat} vs. $t$  in the absence/presence of an eclipse attack on the blockchain network (100 simulations). We used the projected adjacency matrices~\eqref{eq:processed-adjacency-matrix} of dimension 100 to compute $\scaledTestStat_\numData( t)$. 
  When there's an eclipse attack, the peak of the scaled test statistic is well above the $0.95$ quantile of the distribution $\brownianBridge^2(t)=q_{0.95}=9.05$ (Theorem~\ref{thm:null-hypothesis-results}). Moreover, the peak of the scaled test statistic gives the onset of the eclipse attack.
  Therefore, using the processed adjacency matrices decreases the computational cost of the detector while preserving the test statistic (See Sec.\ref{sec:comparison-test-statistic} of the supplementary material for a numerical example comparing the test statistics computed using the original and projected adjacency matrices).}
		\label{fig:test-statistic-jl100-adjacency-matrix}
	\end{figure*}
	In this section, we illustrate  our eclipse attack detection algorithm (Algorithm~\ref{alg:change-detection}) on a simulated
 dataset.
 Sec.\ref{sec:simulation-setup} describes the process of generating a simulated dataset using the eclipse attack model in Definition~\ref{def:eclipse-attack}. Sec.\ref{sec:numerical-proposed-detector} studies the performance of the proposed eclipse detector when applied to the simulated dataset. Sec.\ref{sec:numerical-proposed-detector-noisy} plots the ROC curve for the proposed eclipse detector on a noisy dataset. Sec.\ref{sec:comparison-test-statistic} studies the effect of projected adjacency matrices on the false alarm rate of the detector. Sec.\ref{sec:numerical-comparison-detector} compares the proposed eclipse attack detector against an eclipse attack detector based on the RFM. Sec.\ref{sec:numerical-rfc-regression} implements a RFM based regressor to estimate the onset of the eclipse attack. Sec.\ref{sec:numerical-rfc-sensitivity} studies the sensitivity of the RFM based detector to variations in the training dataset.\footnote{All numerical examples use Matlab. Our source codes are available in the Sec.\ref{sec:source-code} of the supplementary material.} 
    \subsection{Simulation Setup}
    \label{sec:simulation-setup}
	We use Definition~\ref{def:eclipse-attack} to generate a simulated dataset to illustrate the performance of the eclipse attack detector (Algorithm~\ref{alg:change-detection}). Our dataset represents a large-dimensional\footnote{The number of elements in the adjacency matrix for the BCN is $10^4$.} blockchain network with 100 users; it consists of a sequence of $1000$ adjacency matrices~\eqref{eq:adjacency-matrix} for the BCNs. In the absence of an eclipse attack, each blockchain user randomly and uniformly selects five neighbors. However, to simulate an eclipse attack, we introduced one victim user and two malicious users into the blockchain. The malicious users always include the victim user as one of their neighbors and other four neighbors are chosen uniformly at random. Equivalently, $\distributionEclipseAttackAbsent,\distributionEclipseAttackPresent$ in Definition~\ref{def:eclipse-attack} are given by
	\begin{align*}    
	\distributionEclipseAttackAbsent(\adjacencyMatrix_{\graph}(i,j)=1)&=\frac{5}{100}\\
	\distributionEclipseAttackPresent(\adjacencyMatrix_{\graph}(i,j)=1)&=\begin{cases}
	1,\quad i=1, j=99,100\\
	\frac{4}{99}, \quad i\neq1, j=99,100\\
	\frac{5}{100},\quad \text{otherwise}
	\end{cases}
	\end{align*}
	Here, the blockchain user with index 1 is the victim, and those with indexes 99 and 100 are the attackers. 
We assume that the nodes in the graph are labeled in descending order of their computation power. Since eclipse attacks target users with high computational power, our numerical examples focus on the first four rows of the adjacency matrix to reduce computational cost. 
 Henceforth, with an abuse of the notation, $\adjacencyMatrix_\graph$ refers to the first four rows of the adjacency matrix. 

    \subsection{Numerical Examples for the Proposed Eclipse Attack Detector}
    \label{sec:numerical-proposed-detector}
 We employ Algorithm~\ref{alg:change-detection} to detect an eclipse attack on the blockchain network. In step 1, we used the projected adjacency matrices~\eqref{eq:processed-adjacency-matrix} with the number of elements equal to 100. In step 2, we used the projected adjacency matrices to compute the scaled test statistic $\scaledTestStat_\numData( t)$~\eqref{eq:scaled-test-stat}. 
 In step 3, we set a significance level of 0.05 for rejecting the null hypothesis $H_0$. We computed the $q_{0.95}$ quantile of the distribution $\max_{t\in[\boundary,1-\boundary]}\brownianBridge^2(t)$ to be 9.05. Fig.~\ref{fig:test-statistic-jl100-adjacency-matrix} plots the scaled test statistic $\scaledTestStat_\numData( t)$ defined in~\eqref{eq:scaled-test-stat} for  both the absence and presence of an eclipse attack. In the presence of an eclipse attack, the peak of the test statistic surpasses the threshold $q_{0.95}$. Moreover, the peak of the scaled test statistic gives the onset of the eclipse attack (Theorem~\ref{thm:alternate-hypothesis}). 
 


    \subsection{ROC Curves for the Proposed Eclipse Attack Detector}
    \label{sec:numerical-proposed-detector-noisy}
    In this section, we investigate how the signal-to-noise ratio (SNR) of the dataset impacts the performance of the proposed eclipse attack detector (Algorithm~\ref{alg:change-detection}). We added noise in the adjacency matrix as follows:
    \begin{align}
        \label{eq:snr}
        Y=X\wedge N,\quad N=\indicator\left\{U>\operatorname{SNR}^{-1}\right\}
    \end{align}
    Here, $X, Y$ denote the noise-free and noisy adjacency matrix, respectively; $U$ denotes a uniform random variable on $[0,1]$; and $\wedge$ denotes the logical and operator.
    This noise simulates scenarios where the network monitor misses communication between two nodes.     
    \begin{figure}[ht]
        \centering
        \includegraphics[scale=0.4
        ]{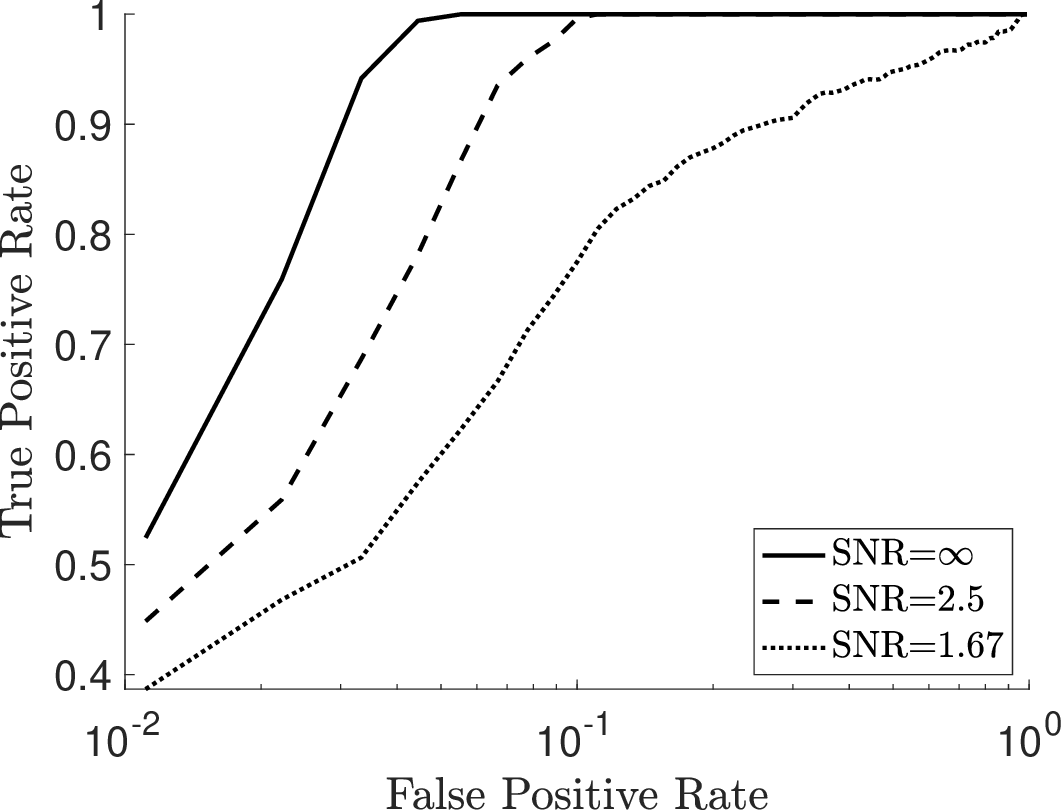}
        \caption{ROC curve of the proposed eclipse attack detector for various SNR values~\eqref{eq:snr}. As observed, the detector performs well with noisy datasets.}
        \label{fig:roc_proposed_snr}
    \end{figure}
    Fig.~\ref{fig:roc_proposed_snr} displays the ROC
     curve~\cite{2004:VB-et-al} for the proposed eclipse attack detector with various SNR values. As observed, the eclipse attack detector is robust to noise.

     \subsection{Comparison of the Test Statistic Computed using Original and Projected Adjacency Matrices}
    \label{sec:comparison-test-statistic}
    \begin{figure*}[t]
        \begin{subfigure}[t]{0.45\textwidth}
            \centering
        \includegraphics[scale=0.42]{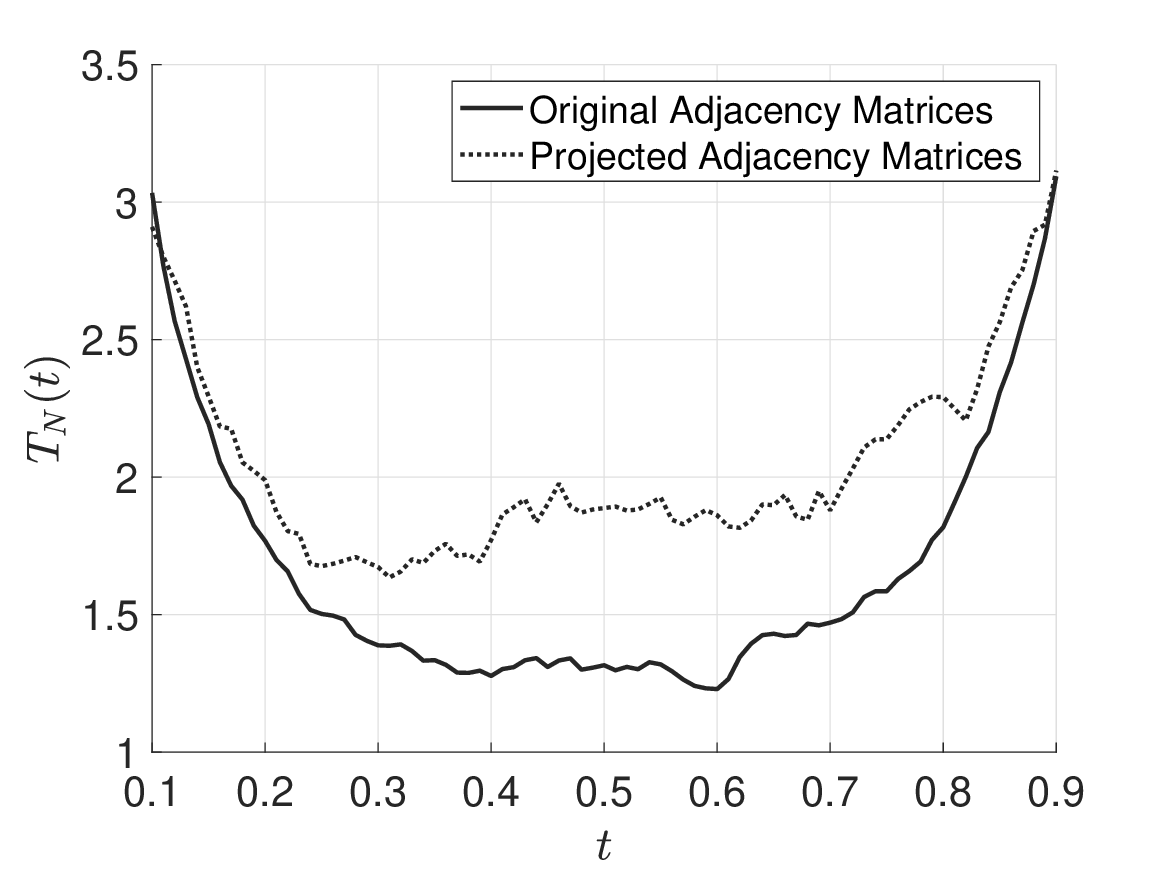}
        \caption{Absence of an eclipse attack on the blockchain network.}
        \label{fig:compare-test-stat-no-attack}
        \end{subfigure}
        \begin{subfigure}[t]{0.45\textwidth}
            \centering
        \includegraphics[scale=0.42]{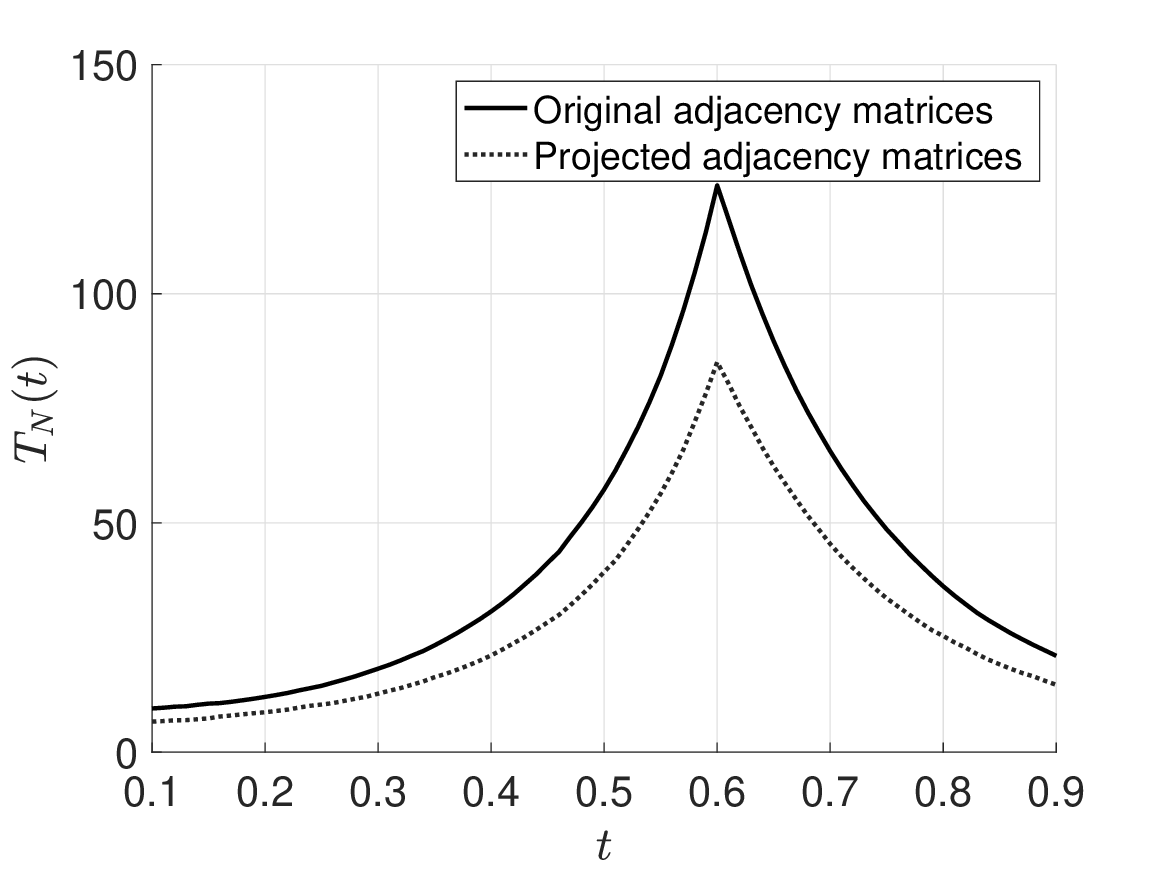}
        \caption{Presence of an eclipse attack on the blockchain network.}
        \label{fig:compare-test-stat-attack}
        \end{subfigure}        
        \caption{Comparison of the scaled test statistic $\scaledTestStat_\numData(t)$ computed using original and projected adjacency matrices. The scaled test statistic is averaged over 100 simulations. As in Sec.\ref{sec:numerical-proposed-detector}, we use the first four rows of the adjacency matrices. Therefore, the number of elements in the original adjacency matrix is 400. We used the JL lemma to obtain the projected adjacency matrices of dimension 100. As observed, the computing the scaled test statistic using the projected adjacency matrices leads to higher false positive and false negative alarm rate.}
    \end{figure*}
    
    Recall that in Theorem~\ref{thm:jl-lemma-test-statistic}, we showed that using the projected adjacency matrices to compute the scaled test statistic leads to a higher false positive alarm rate. Fig.~\ref{fig:compare-test-stat-no-attack} illustrates the impact of projected adjacency matrices on the false alarm rate.  As in Sec.\ref{sec:numerical-proposed-detector}, we use the first four rows of the adjacency matrices. Therefore, number of elements in the original adjacency matrix is 400. We used the JL lemma to obtain the projected adjacency matrices of dimension 100. 
    
    Moreover, in Fig.~\ref{fig:compare-test-stat-attack}, we show using a numerical example that computing the scaled test statistics using the projected adjacency matrices leads to higher false negative rate.

    The two numerical examples justifies the heuristic that the JL lemma approximately preserves the test statistic.
    
    \subsection{Comparison of the Proposed Eclipse Attack Detector with a RFM based  Detector}
    \label{sec:numerical-comparison-detector}
    This section compares the performance of our proposed eclipse attack detector with a RFM~\cite{2019:JS-et-al} based detector. 

To begin, we trained a random forest classifier to detect an eclipse attack on a blockchain network. The training dataset consisted of 390 data points, each corresponding to a sequence of 1000 adjacency matrices for the BCNs (Sec. \ref{sec:simulation-setup}). The simulated dataset was free from noise. 
If a sequence of BCNs was free from an eclipse attack, it was labeled as `0'; otherwise, it was labeled as `1'. 
Following the training of the random forest classifier, 
we validated its performance on a test dataset of size 287. 
    The accuracy of the RFM based detector\footnote{The RFM based detector requires a separate regressor to detect the onset of the eclipse attack. We  study its performance  in Appendix~\ref{sec:numerical-rfc-regression} of the supplementary material.} and the proposed eclipse attack detector (Algorithm \ref{alg:change-detection}) is summarized in Table~\ref{tab:comparison-detection}. Fig.\ref{fig:roc} plots the ROC curve of the two detectors.

    \begin{figure}[ht]
        \centering
        \includegraphics[scale=0.38]{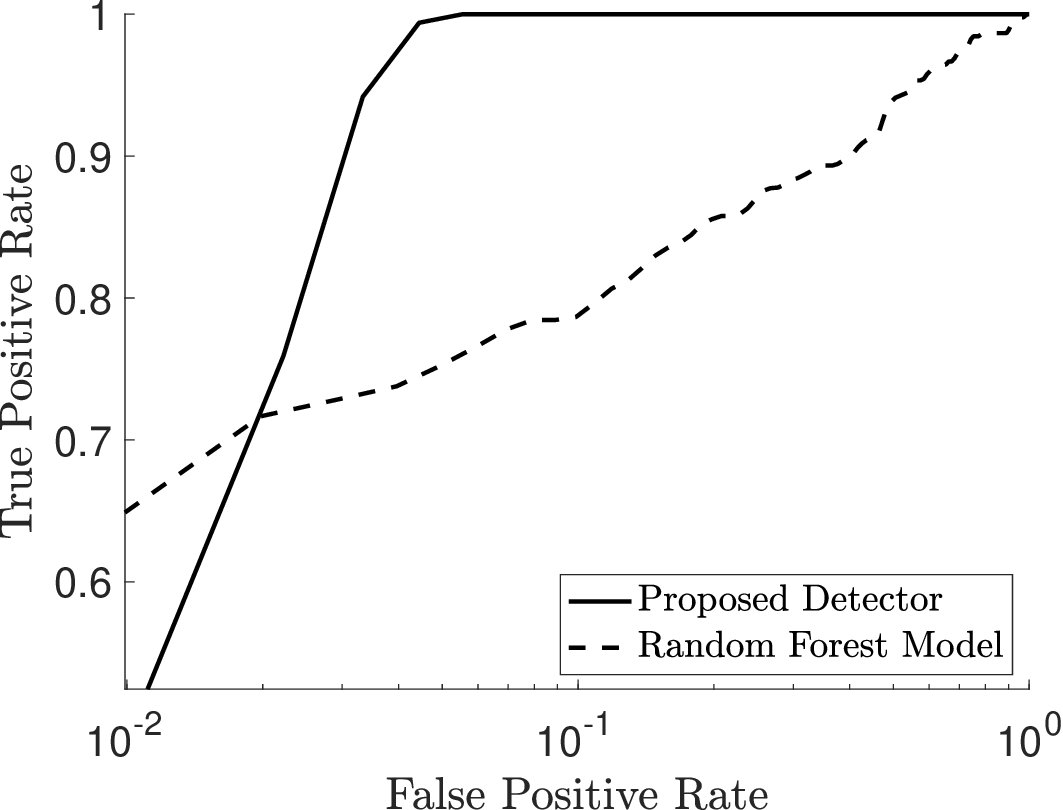}
        \caption{ROC curve of the proposed eclipse attack detector and the RFM based for a dataset with SNR=$\infty$~\eqref{eq:snr}. The proposed detector outperforms the RFM based detector when the false positive rate is high. Note that the RFM based detector requires a training dataset and is sensitive to a training dataset (See Appendix~\ref{sec:numerical-rfc-sensitivity} for a study on sensitivity of the RFM based detector to a training dataset). In contrast, the proposed detector did not require a training dataset.}
        \label{fig:roc}
    \end{figure}
    \begin{table}[ht]
        \centering
        \begin{tabular}{l|l}
            \hline
            \textbf{Detector} & \textbf{Accuracy}\\
            \hline
             Proposed Detector (Algorithm~\ref{alg:change-detection}) &  97.49\%\\
             Random Forest Model & 85.31\%\\
             \hline
        \end{tabular}
        \caption{Accuracy of eclipse attack detectors on a dataset with SNR=$\infty$~\eqref{eq:snr} (100 simulations). 
        }
        \label{tab:comparison-detection}
    \end{table}

	\subsection{RFM based Detector to Estimate the Onset of the Eclipse Attack} 
    \label{sec:numerical-rfc-regression}
    Recall  that in Sec.\ref{sec:numerical-comparison-detector}, we employed a detector based on a RFM to detect an eclipse attack on the blockchain network. In this section, we implement a random forest based regressor to estimate the onset of the eclipse attack  $\changePointExact$ under $H_1$~\eqref{eq:hypothesis}, i.e., presence of an eclipse attack on a blockchain network. Our training dataset comprised 81 data points, each associated with a sequence of $1000$ adjacency matrices denoted as $\adjacencyMatrixDataset$ for the BCNs (Sec. \ref{sec:simulation-setup}). The labels assigned to these data points corresponded to the onset of the eclipse attack. 
    To compare the accuracy in predicting the onset of the eclipse attack, we computed the root mean squared error for both the proposed eclipse attack detector (Algorithm~\ref{alg:change-detection}) and the eclipse attack detector based on the RFM. The results are summarized in Table~\ref{tab:comparison-start-point}.
    \begin{table}[ht]
        \centering
        \begin{tabular}{l|l}
            \hline
            \textbf{Detector} & \textbf{RMSE}\\
            \hline
             Proposed Detector (Algorithm~\ref{alg:change-detection}) &  1.55\\
             Random Forest Model & 38.63\\
             \hline
        \end{tabular}
        \caption{Comparison of root mean squared error (RMSE) in estimating the onset of the eclipse attack on a blockchain network. Our test dataset consisted of 83 data points, each corresponding to a sequence of 1000 adjacency matrices for the BCNs (Sec. \ref{sec:simulation-setup}). The RMSE values were averaged over 5 runs.}
        \label{tab:comparison-start-point}
    \end{table}
    
    The proposed eclipse attack detector outperforms the eclipse attack detector based on the RFM without requiring a training dataset.
    
    \subsection{Sensitivity of the RFM based Detector to Training Dataset}
    \label{sec:numerical-rfc-sensitivity}
    Recall that in Sec.\ref{sec:numerical-comparison-detector}, we implemented a RFM based detector to detect an eclipse attack on the blockchain network. In this section, we study the sensitivity of the eclipse attack detector, based on the RFM, to variations in the training dataset. To investigate this, we generated a training dataset and a test dataset consisting of 390 and 287 data points, respectively. The procedure for generating the dataset is outlined in Sec.\ref{sec:simulation-setup}. The primary distinction between the training and test datasets lies in the number of malicious users. Specifically, the training dataset was designed with 4 malicious users, while the test dataset was configured to include only 2 malicious users.

    Once we trained the random forest classifier, we validated its performance on the test dataset
    We observed a decrease in overall accuracy to 72.25\%. Consequently, achieving a precise random forest regressor requires careful feature extraction from the dataset, with an emphasis on selecting features that remain consistent with the parameters in the eclipse attack model (Definition~\ref{def:eclipse-attack}).

 To summarize, we used a simulated dataset to test the proposed eclipse attack detector (Algorithm~\ref{alg:change-detection}). 
 We also compared the proposed detector with an eclipse detector based on the RFM. 
 Our model stood out by concurrently addressing the two aspects of eclipse attack detection: 1) detecting the presence of an eclipse attack, and 2) estimating the onset of the eclipse attack. Moreover, the proposed eclipse attack detector did not require a training dataset.

	\section{Conclusion}
	This paper  addressed the problem of detecting an eclipse attack on a blockchain network by designing a non-parametric change detection algorithm. In an eclipse attack, malicious users isolate a victim user, disrupting
their ability to reach a consensus with the rest of the
network.
Our eclipse attack detection approach involved estimating changes in the Fréchet mean and variance of the BCN. We showed that the test statistic for the proposed eclipse attack detector weakly converges to a Brownian bridge process. This allowed us to quantify the false alarm rate of the detector. 
 The proposed statistical detector can be implemented as a smart contract on top of the blockchain to mitigate the impact of an eclipse attack. 
 Finally, we used ROC curves to characterize the performance of the proposed eclipse attack detector and the RFM based detector.
It is also worthwhile exploring detection of jump Markov dynamics and the resulting weak convergent statistic; see \cite{YIK09}

	In future work, we will explore: (1) detecting an eclipse attack on a blockchain network with time-varying blockchain users, (2) theoretical bounds on the accuracy of the test statistic when the BCNs are observed in noise, (3) refining the proposed test statistic to effectively detect an eclipse attack near endpoints, and (4) generalizing the change detection algorithm to address time-varying eclipse attack strategies. These extensions will improve the applicability and effectiveness of the proposed eclipse attack detection algorithm.
	\subsubsection*{Acknowledgments}
	This research was supported in part by the U.S. Army Research Office
	grant W911NF-21-1-0093, National Science Foundation grant CCF-2112457, and the Army Research Laboratory under Cooperative Agreement Number W911NF-23-2-0124. 

\bibliography{blockchain}
\bibliographystyle{IEEEtran}

\appendices
    
	\section{Proof of Theorem~\ref{thm:null-hypothesis-results} in Sec.\ref{sec:analysis-algorithm}}
        \label{sec:proof-null-hypothesis-results}
        The outline of the proof is as follows. Observe that 
        \begin{align*}
            \numData\testStat_\numData(\numData t)&=\numData\testStat_\numData^A(\numData t)+\numData\testStat_\numData^B(\numData t)\\
            \numData\testStat_\numData^A(\numData t)&:=\frac{\numData t(1-t)}{\hat{\sigma}^2}\left(\hat{V}_{\numData t}-\hat{V}_{\numData(1-t)}\right)^2\\
		\numData\testStat_\numData^B(\numData t)&:=\frac{\numData t(1-t)}{\hat{\sigma}^2}\left(\hat{V}_{\numData t}^C-\hat{V}_{\numData t}+\right.\\&\left.\hat{V}_{\numData (1-t)}^C-\hat{V}_{\numData (1-t)}\right)^2
        \end{align*}
        \textit{Step 1:} Show that  $\numData\testStat_\numData^A(\numData t)\xrightarrow{w}\brownianBridge^2(t),t\in[\boundary,1-\boundary]$.\\
        To show this, first define $Z_\numData(t)=\sqrt{\numData\testStat_\numData^A(\numData t)}$. Then, for $t_0=\boundary\leq t_1\leq\ldots\leq t_k\leq 1=t_{k+1}$ show that $$(Z_\numData(t_1),Z_\numData(t_2),\ldots,Z_\numData(t_{k})\xrightarrow{w}\scrN(0,\Sigma)$$ where, \begin{align*}\Sigma_{t_1,t_2}&=\indicator(t_1=t_2)\\&+[t_1(1-t_2)/t_2(1-t_1)]^{\frac{1}{2}}\indicator(t_1\neq t_2),\;t_1\leq t_2\end{align*} Finally, show that $Z_\numData(t)$ is asymptotically equicontinuous in probability. Step 1 follows from Donsker's theorem.
        \\
        \textit{Step 2:} Show that  $\numData\testStat_\numData^B(\numData t)\xrightarrow{w}0$ by proving the consistency of estimators under $H_0$.\\
        Theorem~\ref{thm:null-hypothesis-results} follows from combining Step 1 and Step 2 using Slutsky's theorem. Refer~\cite{2020:PD-HM} for the detailed proof of Theorem~\ref{thm:null-hypothesis-results}. 
        \section{Proof of Theorem~\ref{thm:alternate-hypothesis} in Sec.\ref{sec:analysis-algorithm}.}
        \label{sec:proof-alternative-hypothesis}
        \begin{proof}
            The outline of the proof is as follows. Consider two cases: (1) $\changePointGuess\geq\changePointExact$ and (2) $\changePointGuess\leq\changePointExact$. The proof of case (2) is similar to case (1). For case (1), one can show that
        \begin{align*}
            \testStat(\numData t)&\leq \frac{t(1-t)}{\sigma^2}\left(\max(\alpha^2(V_1-V_2)^2,\right.\\
            &\left.(\alpha(V_1-V_2)+\min(\alpha \Delta_1,(1-\alpha)\Delta_2))^2)\right)\\
            &\leq \alpha(V_1-V_2)^2+\alpha(\Delta_1+\Delta_2)^2=\testStat(\numData\changePointExact)
        \end{align*}
        where, 
        \begin{align*}
        \alpha&=\frac{\changePointExact}{\changePointGuess}\\
        \Delta_1&=\mathbb{E}_{\distributionEclipseAttackAbsent}\left[\metric^2\left(\adjacencyMatrix_\graph, \mu_2\right)\right]-\mathbb{E}_{\distributionEclipseAttackAbsent}\left[d^2\left(\adjacencyMatrix_\graph, \mu_1\right)\right]\\
        \Delta_2&=\mathbb{E}_{\distributionEclipseAttackPresent}\left[d^2\left(\adjacencyMatrix_\graph, \mu_1\right)\right]-\mathbb{E}_{\distributionEclipseAttackPresent}\left[d^2\left(\adjacencyMatrix_\graph, \mu_2\right)\right]\\
        \sigma^2&=\changePointExact \mathbb{E}_{\distributionEclipseAttackAbsent}\left[\metric^4(\adjacencyMatrix_\graph, \tilde{\mu})\right]+(1-\changePointExact) \mathbb{E}_{\distributionEclipseAttackPresent}\left[\metric^4(\adjacencyMatrix_\graph, \tilde{\mu})\right]-\tilde{V}^2\\
        \tilde{\mu}&=\arg\min_{\omega \in \graphSpace} \left\{\changePointExact \mathbb{E}_{\distributionEclipseAttackAbsent}\left[\metric^2(\adjacencyMatrix_\graph, \adjacencyMatrix_\omega)\right]\right.\\
        &\left.+(1-\changePointExact) \mathbb{E}_{\distributionEclipseAttackPresent}\left[\metric^2(\adjacencyMatrix_\graph, \adjacencyMatrix_\omega)\right]\right\} \\ \tilde{V}&=\tau \mathbb{E}_{\distributionEclipseAttackAbsent}\left[\metric^2(\adjacencyMatrix_\graph, \tilde{\mu})\right]+(1-\tau) \mathbb{E}_{\distributionEclipseAttackPresent}\left[\metric^2(\adjacencyMatrix_\graph, \tilde{\mu})\right]\\
        \mu_i&=\argmin_{\omega\in\graphSpace}\mathbb{E}_{P_i}\left[\metric^2(\adjacencyMatrix_\graph,\adjacencyMatrix_\omega)\right],\;i=1,2\\
        V_i&=\min_{\omega\in\graphSpace}\mathbb{E}_{P_i}\left[\metric^2(\adjacencyMatrix_\graph,\adjacencyMatrix_\omega)\right],\;i=1,2
        \end{align*}
        The second inequality is obtained from the first inequality by considering multiple sub-cases.
        Refer~\cite{2020:PD-HM} for the detailed proof of Theorem~\ref{thm:alternate-hypothesis}. 
        \end{proof}
	\section{Proof of Theorem~\ref{thm:jl-lemma-test-statistic} in Sec.\ref{sec:analysis-algorithm}}
	\label{sec:proof-jl-lemma-test-statistic}
	\begin{proof}
		To prove Theorem~\ref{thm:jl-lemma-test-statistic}, we first derive an upper and a lower bound on the variance of the projected adjacency matrices $(\processedAdjacencyMatrix_{\graph_i})_i$~\eqref{def:processed-adjacency-matrix} in terms of the variance of the adjacency matrices of the BCNs $(\adjacencyMatrix_{\graph_i})_i$~\eqref{eq:adjacency-matrix}. Then, we compute the value of the test statistic $\tilde{\testStat}_\numData(\numData t)$ and $\testStat_\numData(\numData t)$ for $\numData\rightarrow\infty$ under $H_0$.
		
		\textit{Step1: Comparing the variances:}     
		Using triangle inequality, Lemma~\ref{lemma:jl} and the fact that $\argmax_{\lambda}\Expect{\|\adjacencyMatrix-\lambda\|^2}=\Expect{\adjacencyMatrix}$, we can compare the variance of the projected adjacency matrices $(\processedAdjacencyMatrix_{\graph_i})_i$~\eqref{def:processed-adjacency-matrix} and the variance of the adjacency matrices of the BCNs $(\adjacencyMatrix_{\graph_i})_i$~\eqref{eq:adjacency-matrix}. Let $\alpha$ be an adjacency matrix s.t. $f(\alpha)=\Expect{\processedAdjacencyMatrix_\graph}$. We obtain
		\begin{align*}
		(1-\epsilon)\left\| \adjacencyMatrix_{\graph_i}-\alpha\right\|^2
		&\leq\left\| \processedAdjacencyMatrix_{\graph_i}-\Expect{\processedAdjacencyMatrix_{\graph}}\right\|^2\\
		\Rightarrow
		(1-\epsilon)\sum_i\left\| \adjacencyMatrix_{\graph_i}-\alpha\right\|^2
		&\leq\sum_i\left\| \processedAdjacencyMatrix_{\graph_i}-\Expect{\processedAdjacencyMatrix_{\graph}}\right\|^2\\
		\Rightarrow
		(1-\epsilon)\sum_i\left\| \adjacencyMatrix_{\graph_i}-\Expect{\adjacencyMatrix_{\graph}}\right\|^2
		&\leq\sum_i\left\| \processedAdjacencyMatrix_{\graph_i}-\Expect{\processedAdjacencyMatrix_{\graph}}\right\|^2
		\end{align*}
		Let $\beta$ be such that the linear map obtained from the JL lemma yields $f(\Expect{\adjacencyMatrix_{\graph}})=\beta$.
		\begin{align*}
		\left\| \processedAdjacencyMatrix_{\graph_i}-\beta\right\|^2
		&\leq(1+\epsilon)\left\| \adjacencyMatrix_{\graph_i}-\Expect{\adjacencyMatrix_{\graph}}\right\|^2\\
		\Rightarrow
		\sum_i\left\| \processedAdjacencyMatrix_{\graph_i}-\beta\right\|^2
		&\leq(1+\epsilon)\sum_i\left\| \adjacencyMatrix_{\graph_i}-\Expect{\adjacencyMatrix_{\graph}}\right\|^2\\
		\Rightarrow
		\sum_i\left\| \processedAdjacencyMatrix_{\graph_i}-\Expect{\processedAdjacencyMatrix_{\graph}}\right\|^2
		&\leq(1+\epsilon)\sum_i\left\| \adjacencyMatrix_{\graph_i}-\Expect{\adjacencyMatrix_{\graph}}\right\|^2
		\end{align*}
		
		\textit{Step 2: Comparing the value of the test statistic: }
		Under $H_0$~\eqref{eq:hypothesis} as $n\rightarrow\infty$, $\hat{V}_{\changePointGuess}\rightarrow V,\,\hat{V}_{\changePointGuess}^C\rightarrow V,\, \hat{V}_{\numData-\changePointGuess}\rightarrow V,\,\hat{V}_{\numData-\changePointGuess}^C\rightarrow V$. Here, the convergence is in probability. This implies $\testStat_\numData(\numData t)=0$. 
        Using the previous inequalities, one obtains
        \begin{align*}
            \tilde{\testStat}_\numData(\numData t)\geq \frac{5\epsilon\, t(1-t)V}{\hat{\sigma}^2}     
        \end{align*}
	\end{proof}
\end{document}